\documentclass[graybox]{svmult}

\usepackage{makeidx}         
\usepackage{graphicx}        
\usepackage{multicol}        
\usepackage[bottom]{footmisc}

\usepackage[english]{babel}
\usepackage{amssymb,amstext,amsmath}
\usepackage[misc]{ifsym}  

\makeindex             

\begin{document}

\title*{Robust Morphometric Analysis based on Landmarks. Applications}

\author{A.~Garc\'{\i}a-P\'erez   \and M.Y.~Cabrero-Ortega}
\institute{A.~Garc\'{\i}a-P\'erez (\Letter )  \at
Departamento de Estad\'{\i}stica, I.O. y C.N., Universidad Nacional de Educaci\'on a Distancia (UNED), Paseo Senda del Rey  9, 28040-Madrid, Spain\\
\email{agar-per@ccia.uned.es}
           \and
M.Y.~Cabrero-Ortega \at
C.A. Universidad Nacional de Educaci\'on a Distancia (UNED)-Madrid, Spain\\
\email{ycabrero@madrid.uned.es}
}

\maketitle

\abstract{Procrustes Analysis is a Morphometric method based on Configurations of Landmarks that estimates the
superimposition parameters by least-squares; for this reason, the procedure is very sensitive to outliers. In the first part of the paper we robustify this technique to classify individuals from a descriptive point of view.
In the literature there are also classical results, based on the normality of the observations, to test whether there are significant differences
between individuals. In the second part of the paper we determine a Von Mises plus Saddlepoint approximation for the tail probability of the Procrustes Statistic when the observations come from a model close to the normal. We conclude the paper with some applications using the  Geographical Information System QGIS. }

\abstract*{Procrustes Analysis is a Morphometric method based on Configurations of Landmarks that estimates the
superimposition parameters by least-squares; for this reason, the procedure is very sensitive to outliers. In the first part of the paper we robustify this technique to classify individuals from a descriptive point of view.
In the literature there are also classical results, based on the normality of the observations, to test whether there are significant differences
between individuals. In the second part of the paper we determine a Von Mises plus Saddlepoint approximation for the tail probability of the Procrustes Statistic when the observations come from a model close to the normal.  We conclude the paper with some applications using the  Geographical Information System QGIS.}

\section{Introduction}
\label{sec:1}

\begin{figure}[b]
\sidecaption
\includegraphics[scale=.4]{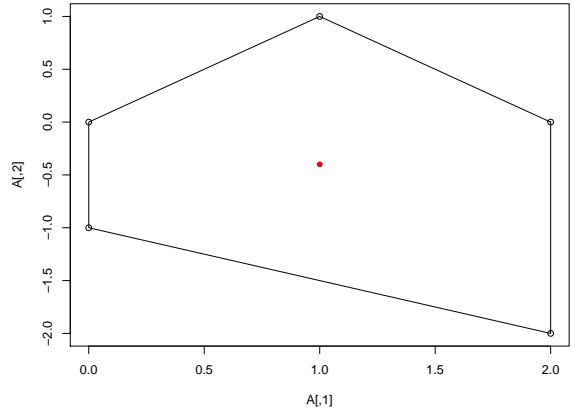}
\caption{Polygon representation of a Configuration with 5 landmarks}
\label{fig:1}
\end{figure}

This paper is about a robust classification problem of $n$ individuals based on their shapes, i.e., using their geometric information. The usual (classical or robust) methods based on a Multivariate Analysis can not extract all the geometric information from the individuals. For this reason, in recent years morphometrics methods based on Configurations of landmarks have been developed.
A landmark is a peculiar point whose position is common in all the individuals to classify. For instance,
when we classify skulls, the landmarks could be the center of the supraorbital arch, the chin, etc.; or, if we classify projectile points found in an archaeological site, the landmarks  could be the ends of the points.

In all the cases, the mathematical (geometric) information that we obtain from individuals is the $k$ coordinates of their $p$ landmarks, $ l_i =(c_{i1},...,c_{ik})$, $\; i=1,...,p$.

The matrix of landmarks coordinates is called a Configuration. For each individual with $p$ landmarks of dimension $k$ (2 or 3) we shall have a collection of landmark coordinates expressed in  $p\times k$ matrix as

$$ M =
\left(
\begin{array}{ccc}
c_{11} & \cdots & c_{1k} \\
\cdots & \cdots & \cdots \\
c_{p1} & \cdots & c_{pk} \\
\end{array}
\right)
$$

\section{Classical Morphometric Analysis from a Descriptive Point of View}
\label{sec:2}

As we have mentioned before, we shall use the shape of the individuals in their classification. Shape is a property of an object
invariant under scaling, rotation and translation; otherwise, for instance,
an object and itself with double size could be classified into two different groups.

\begin{figure}[b]
\sidecaption
\includegraphics[scale=.4]{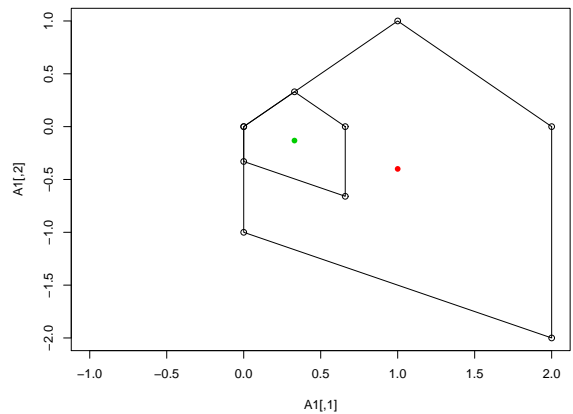}
\caption{Polygons representing a Configuration with 5 landmarks: the original one with red centroid and the scaled
to Centroid-mean Size equal to 1 with green centroid}
\label{fig:2}
\end{figure}

There are many morphometric methods; see for instance  \cite{C} or \cite{DM}. In this paper we shall consider Superimposition Methods; namely, Procrustes Analysis, obtaining the  Procrustes coordinates  with it, adapting the Configurations to a common (local) reference system and matching them  at the common center. For these reasons, a Local Coordinate Reference System is needed and a Geographical Information System will be very useful.

A common graphical representation of a Configuration is a scatter plot of its landmarks coordinates. Joining them with segments we obtaining a polygon as, for instance, in Fig.~\ref{fig:1}, where the landmarks coordinates are the vertices of the polygon.

As we have said above, to classify individuals we have first to remove the effect of  Size (scale),  Location (translation) and  Orientation (rotation) to standardize them and match them in a common center (the centroid of the polygon) in order to make them comparable.

To apply the Procrustes superimposition method  we have to estimate by least-squares the superimposition parameters
$\alpha$, $\beta$ and  ${\boldmath \Gamma}$ (scale,  translation and rotation) in order to minimize the full Procrustes  distance $d_F$ between Configurations  $M_1$ and $M_2$, i.e.,

$$\min  d_F(M_1,M_2) = \min || M_2 - \alpha M_1 {\boldmath \Gamma} - {\boldmath 1}_p \beta' || =$$

$$= \sqrt{trace[(M_2 - \alpha  M_1 \, {\boldmath \Gamma} - {\boldmath 1}_p \beta')' (M_2 - \alpha  M_1 \, {\boldmath \Gamma}- {\boldmath 1}_p \beta')]}
$$

\noindent
where $\alpha$ is a scalar representing the Size,  $\beta$ is a vector of $k$ values corresponding to a Location parameter formed by the centroid coordinates, ${\boldmath 1}_p$ is a column vector of dimension $p \times 1$ and ${\boldmath \Gamma}$ a $k \times k$ square rotation matrix.

The idea that we pursue with this transformation is to match both Configurations, i.e., a superimposition of $M_1$ onto $M_2$.

\subsection{Removing the Size Effect}

The first step we must take in Procrustes Analysis to standardize Configurations is to remove the {\it Size} effect. If, as usual, we consider as center the {\it centroid-mean} of dimension $k$ (sample mean by columns) defined by

$$M_c=\left( M_{c1} , ... , M_{ck} \right)= \left(\frac{1}{p} \sum_{i=1}^p c_{i1} , ... ,   \frac{1}{p} \sum_{i=1}^p c_{ik} \right)$$

\noindent
and easily computed with R as (\cite{RCT})

\begin{verbatim}

> apply(M,2,mean)

\end{verbatim}

\noindent
the {\it Centroid-mean Size} is defined as

\begin{figure}[b]
\sidecaption
\includegraphics[scale=.4]{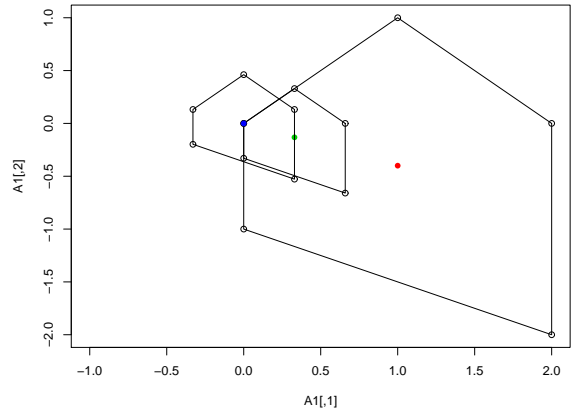}
\caption{Polygons representing a Configuration with 5 landmarks: the original one with red centroid, the scaled to Centroid-mean Size equal to 1, with green centroid, and the centered with respect location and translation with blue centroid}
\label{fig:3}
\end{figure}

$$CS =   \sqrt{\sum_{i=1}^p   d_E^2 (l_i,M_c)} =
\sqrt{\sum_{i=1}^p   \sum_{j=1}^k (c_{ij} - M_{cj})^2} = \sqrt{  \sum_{j=1}^k p \cdot Var(c_{.j})} $$

\noindent
being $d_E^2 (l_i,M_c)$ the square of the Euclidean distance between the $i$th landmark  $l_i$ and the centroid-mean $M_c$. Hence, the Centroid-mean Size depends on the sample variance and so, it will be very sensitive to outliers. This size can be computed as

\begin{verbatim}

> sqrt(sum(apply(M,2,var)*(p-1)))

\end{verbatim}

The coordinates of a scaled Configuration are now calculated dividing the original coordinates by $CS$

$$M_{cs} = \frac{M}{CS}$$

In Fig.~\ref{fig:2} we see the previous Configuration (with red centroid) and the scaled to Centroid-mean Size equal to 1 (the Configuration with green centroid).

\begin{figure}[b]
\sidecaption
\includegraphics[scale=.4]{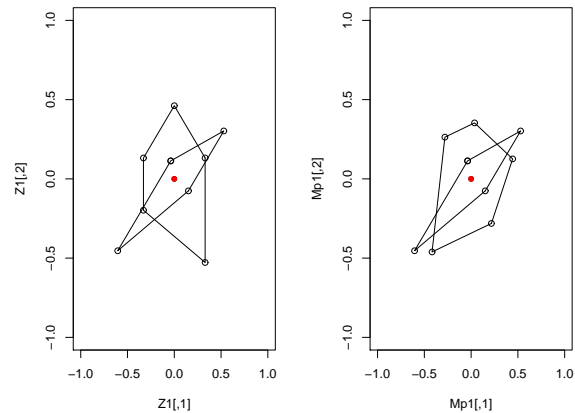}
\caption{Polygon rotated}
\label{fig:4}
\end{figure}

\subsection{Removing Location by Translation}

We remove the Location effect translating the Configuration matrix so that its centroid is the new origin. We do this with the R sentence

\begin{figure}[b]
\sidecaption
\includegraphics[scale=.4]{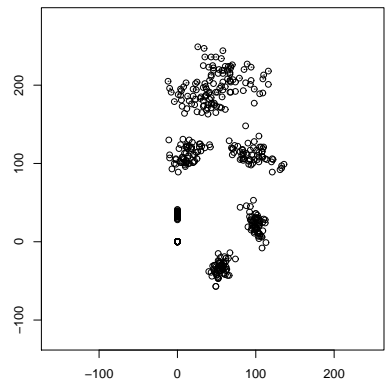}
\caption{Scatter plot of landmarks of Example 1}
\label{fig:5}
\end{figure}

\begin{verbatim}

> scale(M,scale=F)

\end{verbatim}

In Fig.~\ref{fig:3} we have the previous Configurations and the centered one (with blue centroid).

\subsection{Removing Orientation by Rotation}

After the effect of Size and Location have been removed, we estimate (by least-squares) the rotation matrix ${\boldmath \Gamma}$ minimizing the distance between Configurations  $M_1$ and $M_2$, i.e.,

$$\min_{\boldmath \Gamma} || M_2 - M_1 {\boldmath \Gamma} || = \min_{\boldmath \Gamma} \sqrt{trace((M_2 - M_1 \, {\boldmath \Gamma})' (M_2 - M_1 \, {\boldmath \Gamma}))}  $$

\noindent
where  ${\boldmath \Gamma}$ is a $k \times k$ square rotation matrix, a matrix that must be determined in order to maximize the correlation between the two
sets of landmarks, i.e., to minimize the distance between landmarks.  More precisely:

If $M_1$ and $M_2$ are two Configurations and  $X_1$ and $X_2$ the corresponding centered Configurations scaled to unit Centroid-mean Size, the (full) Procrustes distance is defined as

$$ d_F(M_1,M_2) = \sqrt{ trace \left( X_2 - \beta X_1 \, {\boldmath \Gamma}\right)' \left( X_2 - \beta X_1 \, {\boldmath \Gamma}\right)}=$$

$$=\sqrt{1-\left( \sum_{j=1}^k \lambda_j\right)^2}$$

\begin{figure}[b]
\sidecaption
\includegraphics[scale=.4]{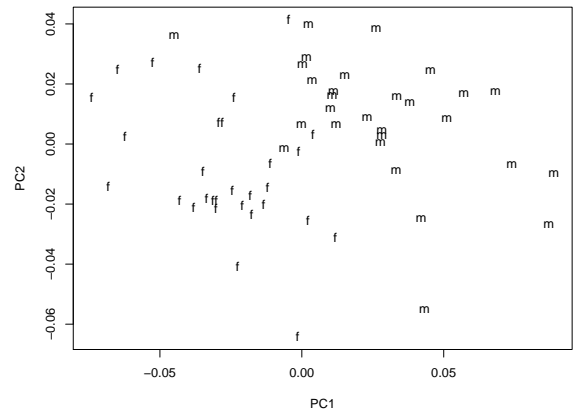}
\caption{Scores of the Principal Component analysis}
\label{fig:6}
\end{figure}

\noindent
where $\lambda_j$ are the diagonal elements of matrix $D$ in the
singular-value decomposition of $X_2' X_1$:

$$ X_2' X_1 = U D V' $$

But in fact, the previous problem is a known mathematical issue: If we have the previous singular-value decomposition of $X_2' X_1$, the rotation matrix we are looking for is

$${\boldmath \Gamma} = V U'$$

In Fig.~\ref{fig:4} we have two Configurations, before and after rotated one of them, according to the previous method.

\subsection{More than two Configurations (Generalized Procrustes Analysis)}

In the previous sections we have done, in three steps, what is called a classical {\it Partial Procrustes Analysis} because we have compared, from a descriptive point of view, two Configurations.

If we have more than two Configurations we have to do what is called a {\it Generalized Procrustes Analysis} in three steps:

\begin{figure}[b]
\sidecaption
\includegraphics[scale=.4]{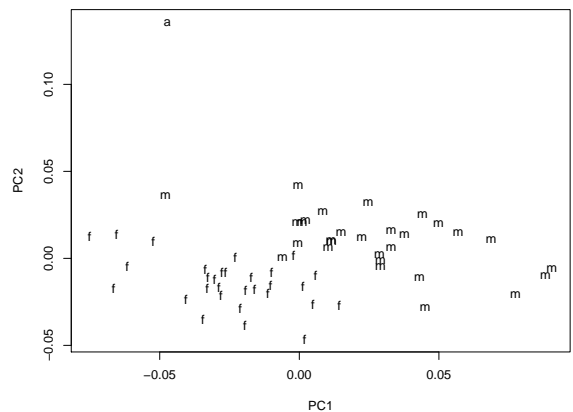}
\caption{Scores of the Principal Component analysis with the outlier {\tt a}}
\label{fig:7}
\end{figure}

\begin{enumerate}

\item All the Configurations must be standardized, i.e., their centroids matched at a common origin and scaled  to unit size.

\item We have to define a consensus or average Configuration  as a reference, called {\it Mean Shape}
because, in fact, it is the sample mean of all the Configurations, by vertices, i.e., the mean of the vertices $(i,j)$ (homologous coordinates) for the $n$ Configurations. The Mean Shape Configuration can be computed as

\begin{verbatim}

> apply(M,c(1,2),mean)

\end{verbatim}

\item Finally, we  perform a (partial) Procrustes superimposition  between the Mean Shape and each Configuration.

\end{enumerate}

\subsection{Configuration Projection onto the Tangent Space}

 The shape space defined by the previous Procrustes superimposition method is non-Euclidean and corresponds to a curved surface. This means that the distance between two landmarks is not the length of the segment joining them and hence, we cannot apply traditional statistics to the Procrustes coordinates  of the $n$ individuals.

 From a convenient point of view, the $n$ individuals  are {\sl aligned} in the $n \times kp$ matrix $X$ and then projected onto the (Kendall) tangent space, where the vectorized  Mean Shape $x_m$ (i.e., a vector of dimension $k p \times 1$) is the contact point between spaces. The projected (tangent) coordinates are obtained in a matrix $X^*$ as

 $$X^* = X(I_{kp} - x'_m x_m)$$

\noindent
where $I_{kp}$ is the $\; \; kp \times kp \; \;$ identity matrix.

Then, we can apply the usual statistical techniques to these projected coordinates, for example, classifying the resulting observations with the  scores of their Principal Components.

\begin{example}
\label{exa:1}
In paper \cite{HD}, 59 gorilla skulls were considered. We know, in the example but not in a real case, that 30 of them are female
 and 29 male. In their paper, 8 landmarks were considered. If we represent these data in a scatter plot we obtain Fig.~\ref{fig:5} where no apparent classification between males and females is observed.

If we make the four previous steps of the Generalized Procrustes analysis and conclude with a  Principal Component analysis of the scores, we obtain
Fig.~\ref{fig:6} where we cannot appreciate the two groups very clearly although the vertical bar at {\tt PC1}=0 is the usual classification rule taken for this example.
\end{example}

\begin{figure}[b]
\sidecaption
\includegraphics[scale=.4]{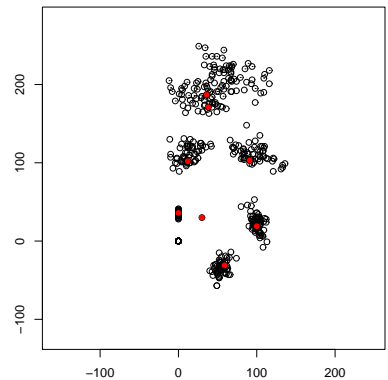}
\caption{Scatter plot of the 59 gorillas plus the outlier {\tt a} in red}
\label{fig:8}
\end{figure}

\section{Robust Morphometric Analysis from a Descriptive Point of View}
\label{sec:3}

Let us consider again the data of Example~\ref{exa:1} plus a Configuration, that we call {\tt a}, for which we replace the coordinates of the third landmark as in the following diagram:

\begin{verbatim}

     [,1] [,2]                           [,1] [,2]
[1,]   36  187                     [1,]    36  187
[2,]   59  -31                     [2,]    59  -31
[3,]    0    0       ---->         [3,]    30   30
[4,]    0   36                 a = [4,]     0   36
[5,]   12  102                     [5,]    12  102
[6,]   38  171                     [6,]    38  171
[7,]   91  103                     [7,]    91  103
[8,]  100   19                     [8,]   100   19

\end{verbatim}

If we give the four previous steps to perform a  Generalized Procrustes analysis, we obtain the classification given in
Fig.~\ref{fig:7} where all the individuals are in one group except outlier {\tt a}.

But in Fig.~\ref{fig:8} we see that {\tt a} is in the bulk of the data and also the mean shape in Fig.~\ref{fig:9}. Hence, no apparent solution is clear.

\begin{figure}[b]
\sidecaption
\includegraphics[scale=.4]{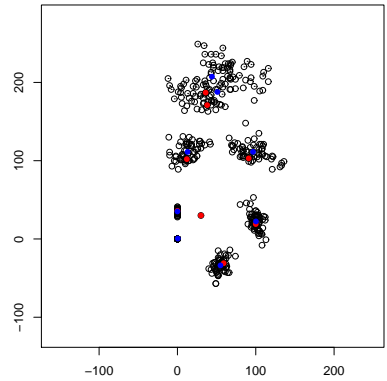}
\caption{Scatter plot of the 59 gorillas plus the outlier {\tt a} in red and the mean shape in blue}
\label{fig:9}
\end{figure}

\subsection{Removing the Size Effect in a Robust Way}

\begin{figure}[b]
\sidecaption
\includegraphics[scale=.4]{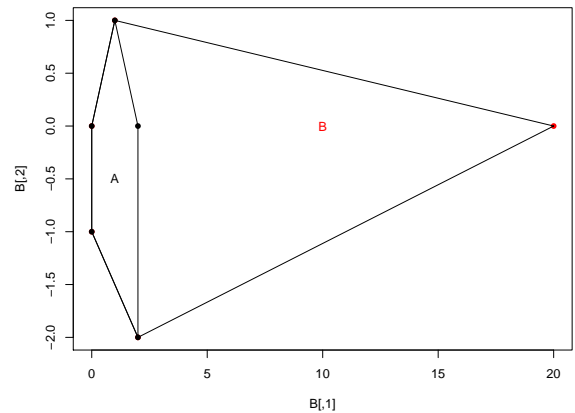}
\caption{Two configurations with a very different size}
\label{fig:10}
\end{figure}

We propose, instead of using the centroid-mean

$$M_c=\left( M_{c1} , ... , M_{ck} \right)= \left(\frac{1}{p} \sum_{i=1}^p c_{i1} , ... ,   \frac{1}{p} \sum_{i=1}^p c_{ik} \right)$$

\noindent
as before, that essentially is a sample mean computed with

\begin{verbatim}

> apply(M,2,mean)

\end{verbatim}

\begin{figure}[b]
\sidecaption
\includegraphics[scale=.4]{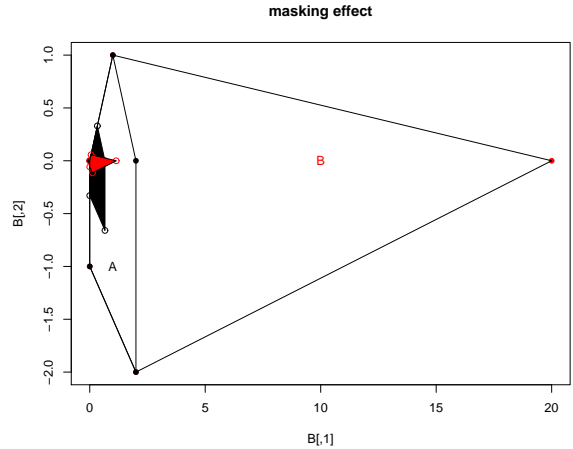}
\caption{Two configurations with a very different size after being standardized with the centroid-mean size}
\label{fig:11}
\end{figure}

\noindent
to use the median (or the trimmed-mean) by columns with the following two R sentences,

\begin{verbatim}
> apply(M,2,median)

> apply(M,2,mean,trim = .2)
\end{verbatim}

\noindent
obtaining in this way a more robust centroid. Now, instead of considering the Centroid-mean Size $CS$ that, as we saw before, is essentially the variance

$$CS =   \sqrt{\sum_{i=1}^p   d_E^2 (l_i,M_c)} =
\sqrt{\sum_{i=1}^p   \sum_{j=1}^k (c_{ij} - M_{cj})^2} = \sqrt{  \sum_{j=1}^k p \cdot Var(c_{.j})} $$

\noindent
computed with

\begin{verbatim}

> sqrt(sum(apply(M,2,var)*(p-1)))

\end{verbatim}

\noindent
we propose to use the Median Absolute Deviation $MAD$, defined as

$$MAD = 1'4826 \, M_e \left\{ \left|X_i - M_e(X_i)\right|  \right\}$$

\noindent
obtaining with it what we call the {\it Centroid-median Size}

$$MS =    \sum_{j=1}^k MAD(c_{.j}) $$

\noindent
computed as

\begin{verbatim}

> sum(apply(M,2,mad))

\end{verbatim}

\noindent
and that satisfies the Size invariance property  $\, MS(a M)  =  a MS(M) \,$ for any positive scalar $a$.

In this way we obtain a robust size measure. For instance, considering the following two configurations {\tt A} and {\tt B}

\begin{verbatim}

> A                           > B
     [,1] [,2]                     [,1] [,2]
[1,]    2    0                [1,]   20    0
[2,]    1    1                [2,]    1    1
[3,]    0    0                [3,]    0    0
[4,]    0   -1                [4,]    0   -1
[5,]    2   -2                [5,]    2   -2

\end{verbatim}

\noindent
that differ in just a wrong digit in the first landmark of Configuration {\tt B},  the classical Centroid-mean Size is very sensitive:

\begin{verbatim}

> sqrt(sum(apply(A,2,var)*4))
[1] 3.03315
> sqrt(sum(apply(B,2,var)*4))
[1] 17.44706

\end{verbatim}

\begin{figure}[b]
\sidecaption
\includegraphics[scale=.4]{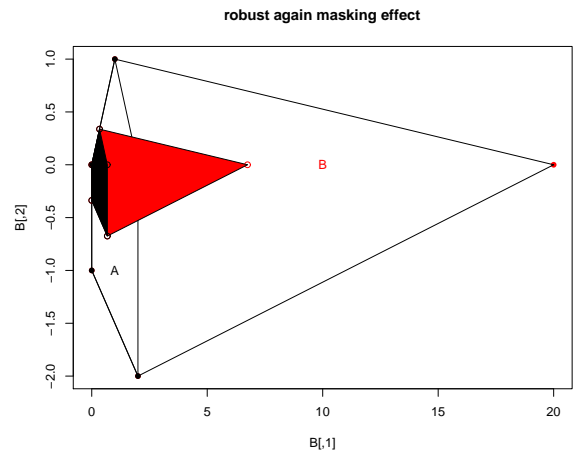}
\caption{Two configurations with a very different size after being standardized with the centroid-median size}
\label{fig:12}
\end{figure}

\noindent
but not the the Centroid-median Size:

\begin{verbatim}

> sum(apply(A,2,mad))
[1] 2.9652
> sum(apply(B,2,mad))
[1] 2.9652

\end{verbatim}

And what is more importante, this new size measure keeps the relative size  of the Configurations avoiding a possible masking effect.
For instance,  in Fig.~\ref{fig:10}, if we divide both Configurations by the classical Centroid-mean Size we obtain Fig.~\ref{fig:11} and they would probably be classified in the same final group. Nevertheless, standardizing them with the new robust Centroid-median Size, we see in  Fig.~\ref{fig:12} that the differences in size  between them, remain.

\begin{figure}[b]
\sidecaption
\includegraphics[scale=.4]{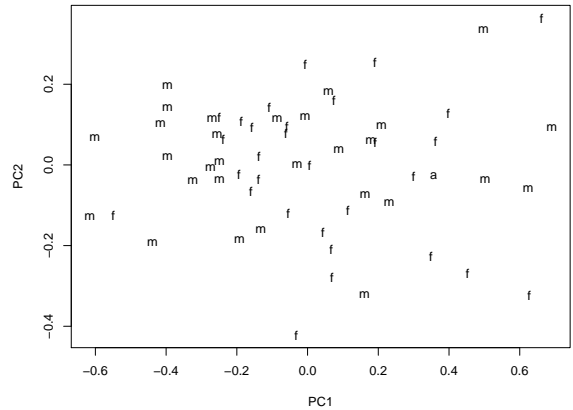}
\caption{Classification of gorillas with the median as {\it mean shape}}
\label{fig:13}
\end{figure}

Hence, instead of dividing the coordinates of the Configuration by the classical Centroid-mean Size $CS$, we  propose to divide the configuration $M$ (the coordinates) by the robust Centroid-median Size to distinguish between individuals in a better way, avoiding a possible masking effect,

$$M_{rs} = \frac{M}{MS}$$

\subsection{Removing Location in a Robust Way}

In the same way as we have removed the location effect in a classical way, translating the Configuration matrix so that its centroid-mean was the new origin, with the sentence

\begin{verbatim}

scale(M,scale=F) = scale(M,scale=F,center=apply(M,2,mean))

\end{verbatim}

\noindent
subtracting the mean of each column to the whole column, in the robust version we subtract the median with the sentence

\begin{figure}[b]
\sidecaption
\includegraphics[scale=.4]{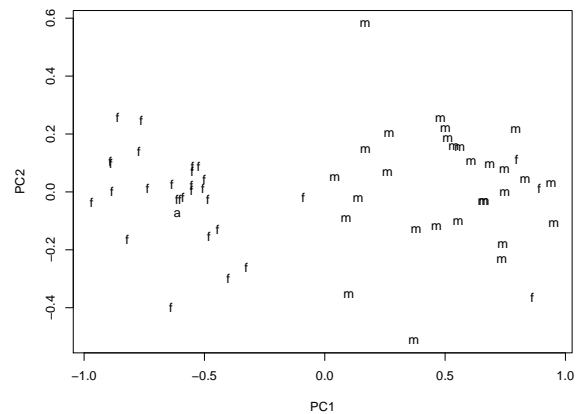}
\caption{Classification of gorillas with the $0'2$-trimmed mean as {\it mean shape}}
\label{fig:14}
\end{figure}

\begin{verbatim}

> scale(M,scale=F,center=apply(M,2,median))

\end{verbatim}

\noindent
being the centroid-median the new origin.

After robustifying with respect scale and location we keep the classical rotation matrix for the robust coordinates. These three steps are in our new R function {\tt rpgpa1} that can be obtained on request from the authors. We conclude the process with the same coordinates projection formula than before.

\subsection{More than two Configurations}

If there are more than two Configurations, the key point in the robustification process is the selection of a robust  {\it mean shape}, that in the classical Morphometric analysis is the sample mean of the Configuration coordinates. In our robust version we propose to choose as {\it mean shape} the
median of the Configuration coordinates, obtaining Fig.~\ref{fig:13} for the gorillas example (after doing a classical Principal Component analysis of the scores).

Considering the $0'2$-trimmed mean as {\it mean shape} we obtain Fig.~\ref{fig:14}.
Finally, considering the $0'1$-trimmed mean as {\it mean shape} we obtain Fig.~\ref{fig:15}.

\begin{figure}[b]
\sidecaption
\includegraphics[scale=.4]{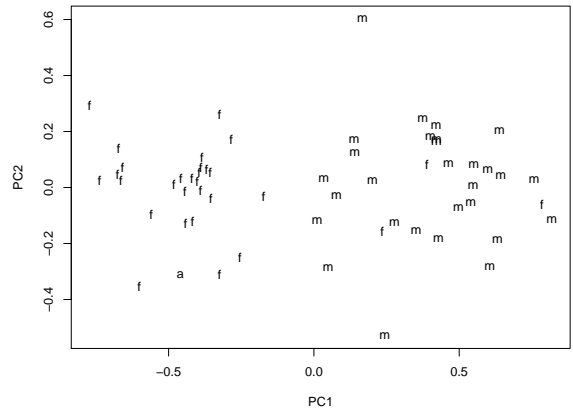}
\caption{Classification of gorillas with the $0'1$-trimmed mean as {\it mean shape}}
\label{fig:15}
\end{figure}

These three options are in
new R function {\tt rpgpa2} that can be composed with {\tt rpgpa1}.

\section{Classical Morphometric Analysis from an Inferential Point of View}
\label{sec:4}

Instead of considering a descriptive morphometric analysis it is more interesting to test if there are significant differences
between  two Configurations. From a classical point of view, we have the following result in
 \cite{LC} and  \cite{S}: If $X_1$ and $X_2$ are two scaled and centered Configurations
 with $p \times k$ landmarks, the {\it Residual Distance} between Configurations $X_1$ and $X_2$  is defined as

$$ || X_2 - X_1 ||^2 = trace\left[(X_2 - X_1)' (X_2 - X_1)\right].  $$

As saw in the previous sections, the $k \times k$ square rotation matrix ${\boldmath \Gamma}$  is determined such that the Procrustes distance between these two Configurations  $X_1$ and $X_2$ (i.e., between landmarks) is minimal

$$\min_{\boldmath \Gamma} || X_2 - X_1 {\boldmath \Gamma} ||^2 = \min_{\boldmath \Gamma} trace\left[(X_2 - X_1 \, {\boldmath \Gamma})' (X_2 - X_1 \, {\boldmath \Gamma})\right] . $$

This minimum (i.e., after matching, i.e., after translation, rotation and scaling) that we obtain  is called the  {\it Procrustes statistic}:

$$ G(X_1,X_2) =
 \min_{{\boldmath \Gamma}}  || X_2 - X_1 {\boldmath \Gamma} ||^2.$$

  Under the null hypothesis $H_0$ that there is no systematic difference between Configurations $X_1$ and $X_2$, i.e., they belong to the same group, or
  more precisely,  that they are of the form

 $$X_2 = X_1 + \eta \, {\boldmath e} $$

 \noindent
 where all the elements of Configuration $\; {\boldmath e} \; $ are i.i.d. $N(0,1)$,  then

 $$G(X_1,X_2)  \approx  \eta^2 \, \chi^2_g$$

 \noindent
 i.e., $G_s(X_1,X_2) =  G(X_1,X_2)/\eta^2  \approx   \chi^2_g \,$,
where $\, g=kp-k(k+1)/2-1 \,$ obtaining so, a way to compute the tail probabilities (p-values) for testing $H_0$. It must be $p> (k+1)/2+1/k$ and obviously an integer.

\section{Robust Morphometric Analysis from an Inferential Point of View}
\label{sec:5}

The standard normality of the landmarks is a very difficult assumption to assume and check. For this reason we shall use robust methods for testing $H_0$ assuming that the $p \times k$ landmarks of $e$ follow, not a standard normal distribution but a contaminated normal model:

$$\frac{X_2 - X_1}{\eta}   \leadsto  (1-\epsilon) N(0,1) + \epsilon N(0,\nu)$$

In this section we are going to compute the tail probabilities (p-values), assuming this contaminated model, using
a VOM+SAD approximation.

We use this scale contaminated normal mixture model because  the Configurations are matched at the common centroid that is the new origin and equal to 0, being the contamination in the scale the source of contamination in the observations.

\subsection{Von Mises Approximations for the p-value of the Procrustes Statistic}

In order to test the null hypothesis $H_0$ that there is no systematic difference between the standardized Configurations $X_1$ and $X_2$, using the
Procrustes statistic  $G_s(X_1,X_2)$  that follow a $\chi^2_g$ distribution under a normal model, we have the following result.

\begin{proposition}
Let $G_s(X_1,X_2)$ be the Procrustes statistic, that follows a $\chi_g^2$ distribution
 when the underlying model is a normal distribution,
$\Phi_{\mu,\sigma}$. If the previous null hypothesis $H_0$ holds, the von Mises (VOM) approximation for the functional tail probability (if $F$ is close to the normal $\Phi_{\mu,\sigma}$) is

$$P_F\{G_s(X_1,X_2) > t \}  \simeq
 g  \int_{-\infty}^\infty
P\{\chi_{g-1}^2 > t- (
\mbox{$ \frac{x-\mu}{\sigma}$}  )^2\}
\, dF(x)
- (g-1)
P\{\chi_{g}^2 > t\}.$$
\end{proposition}

\begin{proof}
The von Mises (VOM) approximation for the functional tail probability is (if $F$ is close to the normal $\Phi_{\mu,\sigma}$)

\begin{equation}
p_g^F = P_F\{G_s(X_1,X_2) > t \}  \simeq  p_g^\Phi + \int \mbox{TAIF}(x;t;\chi^2_g,\Phi_{\mu,\sigma}) \, dF(x)
\label{m*2}
\end{equation}

\noindent
where TAIF is the Tail Area Influence Function defined in \cite{FR}.

Replacing the normal model by the contaminated normal model
 $\Phi^\epsilon = (1-\epsilon) \,
\Phi_{\mu,\sigma} + \epsilon \, \delta_x \;$ and computing the
derivative at $\epsilon=0$ we obtain that

\begin{eqnarray*}
\mbox{TAIF}(x;t;\chi^2_g,\Phi_{\mu,\sigma}) & = & \displaystyle
\left. \frac{\partial}{\partial \epsilon}
P_{\Phi^\epsilon}\{G_s(X_1,X_2) > t\}
\right|_{\epsilon=0} \\
 & = & g P\{\chi_{g-1}^2 > t- (x-\mu)^2/\sigma^2\}
- g P\{\chi_g^2 > t\}
\end{eqnarray*}

\noindent
integrating  now, we obtain the result.
\qed
\end{proof}

Considering a scale contaminated normal (SCN) model

$$(1-\epsilon) N(0,1) + \epsilon N(0,\nu)$$

\noindent
 the VOM approximation is

$$
p_{g}^F \simeq (1 - g \, \epsilon)
P\{\chi_{g}^2 > t\} +
 g \, \epsilon \int_{-\infty}^\infty
P\{\chi_{g-1}^2 > t- x^2 \}
\, d\Phi_{0,\nu}(x).$$

In Table~\ref{Table:1} appear the {\sl Exact} values
(obtained through a simulation of 100.000 samples) and the VOM
approximations when $\epsilon =0'05$, $\nu=2$ and
$g=3$.

\begin{table}
\caption{{\sl Exact} and approximate p-values with $g=3$}
\label{Table:1}
\begin{tabular}{p{2.4cm}p{2.4cm}p{2.4cm}}
\hline\noalign{\smallskip}
$t$ & ``exact" & approximate   \\
\noalign{\smallskip}\hline\noalign{\smallskip}
6  & $0'149$ & $0'148$  \\
8  & $0'077$ & $0'076$    \\
10 & $0'042$ & $0'042$    \\
12 & $0'024$ & $0'025$    \\
14 & $0'016$ & $0'016$    \\
16 & $0'011$ & $0'011$    \\
18 & $0'007$ & $0'008$    \\
\noalign{\smallskip}\hline
\end{tabular}
\end{table}

\

To obtain the previous numerical results we had to deal with
numerical integration. Sometimes, we would
like to have analytic expressions of $p_{g}^F$ to value the effect of contamination $\epsilon$, etc. For this reason and for controlling the relative error of the approximation, in the next section we shall compute the Saddlepoint approximation for the p-value of the Procrustes Statistic.

\subsection{Saddlepoint Approximations for the p-value of the Procrustes Statistic}

Using
 Lugannani and Rice formula, \cite{LR}, for the sample mean of $g$ independent
square normal variables, we obtain the VOM+SAD approximation given in the next result.

\begin{proposition}
Let $G_s(X_1,X_2)$ be the Procrustes statistic, that follows a $\chi_g^2$ distribution
 when the underlying model is a normal distribution,
$\Phi_{\mu,\sigma}$. If the null hypothesis $H_0$ holds, the saddlepoint approximation of the von Mises expansion, VOM+SAD approximation, for the functional tail probability (if $F$ is close to the normal $\Phi_{\mu,\sigma}$) is

\begin{equation}
P_F\left\{ G_s(X_1,X_2) > t \right\} \simeq
P\left\{ \chi^2_g > t \right\} - B + B \, \int_{-\infty}^{\infty}
\frac{ \sqrt{g} }{\sqrt{t}} \,
e^{\frac{(t-g)(x-\mu)^2}{2t\sigma^2}}
\; dF(x)
\label{m*8}
\end{equation}

\noindent
where
$B= \frac{g \, \sqrt{g}}{\sqrt{\pi} \, (t-g) } \,
e^{-(t-g - g\cdot \log (t/g) )/2} $.
\end{proposition}

\begin{proof}
If $G_s(X_1,X_2)$ follows a $\chi_g^2$
distribution, and $Y_1,...,Y_g$ are $g$
independent gamma distributions $\gamma(1/2,1/2)$ with moment generating
function $M$ and cumulant generating function $K=\log M$, it is,
following  \cite{LR}, \cite{D} or \cite{J},

\begin{eqnarray}
P_\Phi\left\{ \frac{G_s(X_1,X_2)}{g} > t \right\}  & =&
P\left\{ \frac{1}{g} \sum_{i=1}^g Y_i > t \right\} \nonumber \\
  &  & \nonumber \\
  & = & 1 - \Phi_s(w) + \phi_s(w) \left\{ \frac{1}{r} - \frac{1}{w} +
O(g^{-3/2}) \right\}
\label{m*7}
\end{eqnarray}

\noindent
where $\Phi_s$ and $\phi_s$ are the cumulative distribution and density
functions of the standard normal distribution.

If $K$ is the cumulant generating function, that is the functional of $\Phi_{\mu,\sigma}$,

$$K(\theta) = \log \int_{-\infty}^\infty e^{\theta \,
(u-\mu)^2/\sigma^2}
\; d\Phi_{\mu,\sigma}(u)$$

\noindent
and  $z_0$ is the (functional) saddlepoint, i.e., it is the solution
of the equation $K'(z_0) = t$,
the functionals that appear in (\ref{m*7}) are

$$w = sign(z_0) \, \sqrt{2 \, g \cdot (z_0 \, t - K(z_0))}
= \sqrt{g} \, sign(z_0) \, \sqrt{2 \, (z_0 \, t - K(z_0))}
:= \sqrt{g}\, w_1
$$

$$r = z_0 \, \sqrt{g \cdot K''(z_0)} = \sqrt{g} \, z_0 \, \sqrt{K''(z_0)}
:= \sqrt{g} \, r_1. $$

As we saw before, the VOM approximation for the tail probability depends on the TAIF. To obtain the TAIF of $\; G_s(X_1,X_2)/g \;$ at $\; \Phi_{\mu,\sigma} \;$
we have to replace the model
$\Phi_{\mu,\sigma}$ by the contaminated model
$\Phi^\epsilon= (1-\epsilon) \Phi_{\mu,\sigma} + \epsilon \, \delta_x$
in all the functionals
in the right side of
(\ref{m*7}) that depend on  $\Phi_{\mu,\sigma}$ , and then to obtain the
derivative at $\epsilon=0$; this process is represented with a dot over
the functional.
 Since  $\; \phi'_s(w)=-\phi_s(w) \, w \; \;$ and
$ \;\phi_s(w)\leq 1  \;$, we obtain that

$$\mbox{TAIF}\left(x;t;\frac{G_s(X_1,X_2)}{g},\Phi_{\mu,\sigma}\right)  =
\left. \frac{\partial}{\partial \epsilon}
P_{\Phi^{\epsilon}} \left\{\frac{G_s(X_1,X_2)}{g} >  t \right\}
\right|_{\epsilon=0} $$

$$= - \phi_s(w)
\stackrel{\bullet}{w} + \phi'_s(w) \stackrel{\bullet}{w}
\left\{ \frac{1}{r} - \frac{1}{w} + O(g^{-3/2}) \right\} + \phi_s(w)
\left\{- \frac{\stackrel{\bullet}{r}}{r^2} +
\frac{\stackrel{\bullet}{w}}{w^2} + O(g^{-3/2}) \right\}
$$

$$=\phi_s(w) \left[
- \frac{ w \, \stackrel{\bullet}{w}}{r}  -
\frac{ \stackrel{\bullet}{r}}{r^2} +
\frac{\stackrel{\bullet}{w}}{w^2} \right]+O(g^{-1}) $$

$$=\phi_s(w) \left[
- \frac{\sqrt{g} \, w_1 \, \sqrt{g} \, \stackrel{\bullet}{w_1}}{
\sqrt{g} \,r_1}  -
\frac{ \sqrt{g}\, \stackrel{\bullet}{r_1}}{g \, r_1^2} +
\frac{\sqrt{g}\,\stackrel{\bullet}{w_1}}{g \, w_1^2}
\right]+O(g^{-1})$$

$$ =  \frac{\phi_s(w)}{r_1} \left[
- \sqrt{g} \cdot w_1 \,
\stackrel{\bullet}{w}_1
\right]+ O(g^{-1/2}) $$

\noindent
because the functionals $w_1,
\stackrel{\bullet}{w}_1, r_1$ and
$\stackrel{\bullet}{r}_1$ do not depend on $g$.
Since

$$\stackrel{\bullet}{w}_1 = sign(z_0) \, \frac{ 2 (
\stackrel{\bullet}{z}_0 \, t - \stackrel{\bullet}{K}(z_0))}{
2 \sqrt{ 2 ( z_0 \, t - K(z_0))}} =
\frac{\stackrel{\bullet}{z}_0 \, t - \stackrel{\bullet}{K}(z_0))}{
w_1}$$

\noindent
it will be

\begin{equation}
\mbox{TAIF}\left(x;t;\frac{G_s(X_1,X_2)}{g},\Phi_{\mu,\sigma}\right)  =
 \frac{\phi_s(w)}{r_1} \sqrt{g}
\left[
\stackrel{\bullet}{K}(z_0) - \stackrel{\bullet}{z}_0 \, t
\right]+ O(g^{-1/2}).
\label{res*1}
\end{equation}

Hence, we have to compute the influence functions $\stackrel{\bullet}{K}(z_0)$ and $\stackrel{\bullet}{z}_0$. To do this, because

$$K'(\theta) = \displaystyle
\frac{ \displaystyle
\int_{-\infty}^\infty e^{\theta \, (u-\mu)^2/\sigma^2} \;
\left(\frac{u-\mu}{\sigma}\right)^2
\; d\Phi_{\mu,\sigma}(u)}{\displaystyle
\int_{-\infty}^\infty e^{\theta \, (u-\mu)^2/\sigma^2} \;
\; d\Phi_{\mu,\sigma}(u)}$$

\noindent
from the saddlepoint equation, $\; K'(z_0)=t \;$, we obtain

$$
\int_{-\infty}^\infty e^{z_0 \, (u-\mu)^2/\sigma^2} \;
\left[\left(\frac{u-\mu}{\sigma}\right)^2 - t\right]
\; d\Phi_{\mu,\sigma}(u)=0.$$

Replacing again the model by the contaminated model
 $\Phi^\epsilon = (1-\epsilon) \,
\Phi_{\mu,\sigma} + \epsilon \, \delta_x \;$ before obtaining the
derivative at $\epsilon=0$, and making the
change of variable $(u-\mu)/\sigma =y$,
we obtain

$$
\stackrel{\bullet}{z}_0 \left[
\int_{-\infty}^\infty e^{z_0 \, y^2} \, y^4
\, d\Phi_{s}(y)
- t
\int_{-\infty}^\infty e^{z_0 \, y^2} \, y^2
\, d\Phi_{s}(y)\right] +
 e^{z_0 \, (x-\mu)^2/\sigma^2} \,
\left[\left(\frac{x-\mu}{\sigma}\right)^2 - t\right]
=0$$

\noindent
i.e.,

$$
\stackrel{\bullet}{z}_0 = \frac{1}{2} \,
t^{-5/2} \,
e^{\frac{(t-1)(x-\mu)^2}{2t\sigma^2}} \,
\left[t - \left(\frac{x-\mu}{\sigma}\right)^2 \right].$$

In a similar way, we obtain that

$$\stackrel{\bullet}{K}(z_0)
= \frac{3}{2} \,
t^{-1/2} \, e^{z_0 \, (x-\mu)^2/\sigma^2} - \frac{1}{2} \,
 t^{-3/2} \, e^{z_0 \, (x-\mu)^2/\sigma^2} \,
\left(\frac{x-\mu}{\sigma}\right)^2 - 1.$$

Also it is

$$r_1=z_0 \, \sqrt{K''(z_0)} = \frac{t-1}{\sqrt{2}} \hspace{1cm}
\mbox{and} \hspace{1cm}
\phi_s(w) = \frac{1}{\sqrt{2\pi}} \, e^{-g \cdot (t-1-\log t)/2}.$$

\noindent
Therefore, from (\ref{res*1}), it will be

$$
\mbox{TAIF}\left(x;t;\frac{G_s(X_1,X_2)}{g},\Phi_{\mu,\sigma}\right) =
A \left(\frac{1}{\sqrt{t}}  \, e^{\frac{(t-1)(x-\mu)^2}{2t\sigma^2}} -
1\right) + O(g^{-1/2})$$

\noindent
where

$$A= \frac{\sqrt{g}}{\sqrt{\pi} \, (t-1) } \,
e^{-g \cdot (t-1 - \log t)/2}. $$

\noindent
From (\ref{m*2}), we obtain now the
VOM+SAD approximation
 for the p-value of the test statistic $G_s(X_1,X_2)/g$,

$$P_F\left\{ \frac{G_s(X_1,X_2)}{g} > t \right\} \simeq
P\left\{ \chi^2_g > g \, t \right\} - A + A \, \int_{-\infty}^{\infty}
\frac{1}{\sqrt{t}} \, e^{\frac{(t-1)(x-\mu)^2}{2t\sigma^2}} \; dF(x)$$

\noindent
and from this, we obtain  the approximation
(\ref{m*8}) for the test statistic $G_s(X_1,X_2)$.
\qed
\end{proof}

If $F$ is the location
contaminated normal mixture (LCN),

$$F= (1-\epsilon) \, N(0,1) + \epsilon
\, N(\theta,1)$$

\noindent
the VOM+SAD approximation is

$$
P_F\left\{ G_s(X_1,X_2) > t \right\} \simeq
P\left\{ \chi^2_g > t \right\} + \epsilon \, B \,
\left[ e^{-(1-t/g)\theta^2/2} -1 \right].$$

In Table~\ref{Table:2} appear the {\sl Exact} values
(obtained through simulation of 100.000 samples), the VOM  and the VOM+SAD
approximations when $\epsilon =0'01$, $\theta=1$ and $g=5$.

\begin{table}
\caption{{\sl Exact} and approximate p-values with $g=5$}
\label{Table:2}
\begin{tabular}{p{2.4cm}p{2.4cm}p{2.4cm}p{2.4cm}}
\hline\noalign{\smallskip}
$t$ & ``exact" & VOM appr. & VOM+SAD appr.   \\
\noalign{\smallskip}\hline\noalign{\smallskip}
9  & $0'1125$ & $0'1129$ & $0'1136$\\
11 & $0'0538$ & $0'0539$ & $0'0545$  \\
13 & $0'0251$ & $0'0249$ & $0'0253$  \\
15 & $0'0114$ & $0'0112$ & $0'0115$  \\
17 & $0'0050$ & $0'0049$ & $0'0051$  \\
19 & $0'0022$ & $0'0022$ & $0'0023$  \\
\noalign{\smallskip}\hline
\end{tabular}
\end{table}

\begin{corollary}
To test the null hypothesis $H_0$ that there is no systematic difference between the standardized  Configurations $X_1$ and $X_2$
with $p \times k$ landmarks
(i.e., they belong to the same group) using the
Procrustes statistic  $G(X_1,X_2)$  and assuming that the error difference between Configurations

$$\frac{X_2 - X_1}{\eta}$$

\noindent
follow a scale contamination normal model $\, (1-\epsilon) N(0,1) + \epsilon N(0,\nu) \,$, the VOM+SAD approximation for the tail probability (p-value) is

 $$P\{ G_s(X_1,X_2) > t \}   \approx  P\{ \chi^2_g > t \} + \epsilon \, \frac{g^{3/2}}{\sqrt{\pi}  (t-g)}
 \left[ \frac{\sqrt{g}}{\sqrt{t-\nu^2(t-g)}} -1 \right] $$

\begin{equation}
\cdot \exp \left\{ -\frac{1}{2} \left( t-g-g \cdot \log\frac{t}{g} \right) \right\}
\label{approx1}
\end{equation}

\noindent
where $g=kp-k(k+1)/2-1$. It must be $p> (k+1)/2+1/k$ and obviously an integer.

Then, if $k=2$, it is $g=2p-4$ and $p > 2$.  And  if $k=3$, it is $g=3p-7$ and $p \geq 3$.
\end{corollary}

\section{Applications}

\begin{figure}[b]
\sidecaption
\includegraphics[scale=.22]{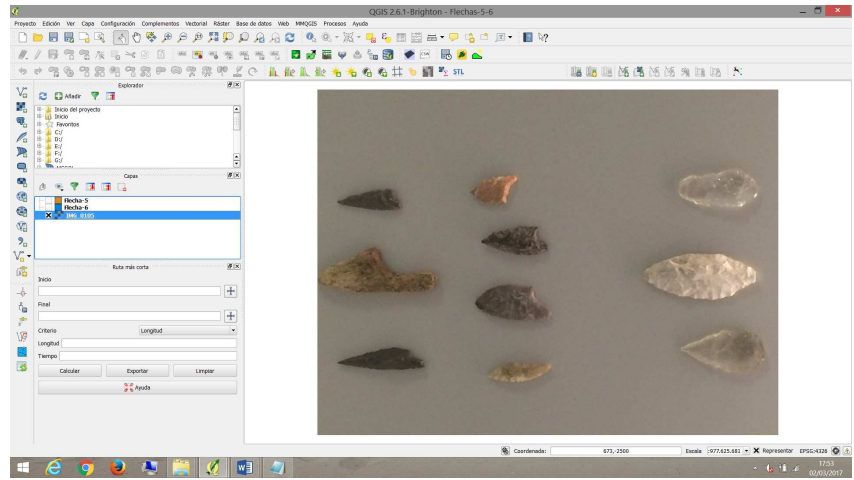}
\caption{Arrows in QGIS}
\label{foto1}
\end{figure}

In this section we are going to consider the following Example in which we make two comparisons using the previous theory.

\begin{example}
\label{exa:2}

We are going to consider two test to check if there are or not significance differences between two arrows of {\sl Notch tips and bay leaves}, of Solutrense period, arrows that were found in caves of Asturias (Spain). We shall make this analysis using a photo of the ``Museo Arqueológico de Asturias'' (Oviedo), included in QGIS as a raster layer, Fig.~\ref{foto1}.

\begin{figure}[b]
\sidecaption
\includegraphics[scale=.22]{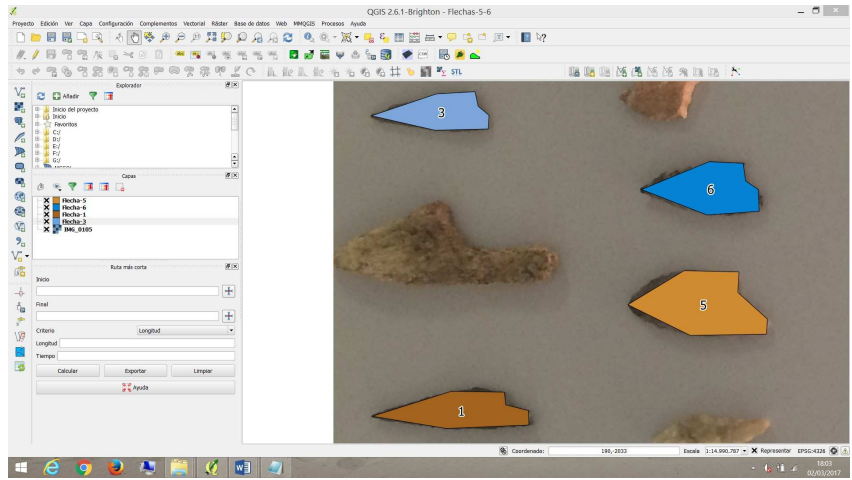}
\caption{Arrows as polygons in QGIS}
\label{foto2}
\end{figure}

In this figure we see large differences among the arrows except in two pairs of arrows: Arrows 1 and 3 and arrows 5 and 6, Fig.~\ref{foto2}.
Hence,  we are going to test the null hypothesis of no significance differences between arrows 1 and 3, and then, with another test, we shall check
the null hypothesis of no significance differences between arrows 5 and 6.

To do this, we first create the polygons in QGIS marking the landmarks with the mouse. We consider $p=7$ landmarks.
Also, with QGIS we export the coordinates of the landmarks,  that are:

\small
\begin{verbatim}

Arrow 3

punta3<-matrix(c(
151.77884,-794.21946,
384.34151,-714.48369,
533.84608,-706.17788,
543.81305,-756.01273,
587.00326,-794.21946,
583.68094,-842.39315,
384.34151,-849.03780),
ncol=2,byrow = T)

Arrow 1

punta1<-matrix(c(
157.59291,-1934.60710,
444.97392,-1841.58204,
640.99102,-1848.22668,
650.95799,-1891.41689,
734.01609,-1917.99548,
735.67725,-1966.16918,
428.36230,-1977.79731),
ncol=2,byrow = T)


Arrow 6

punta6<-matrix(c(
1170.17428,-1072.54821,
1423.18465,-971.34406,
1550.56229,-974.83386,
1557.54188,-1039.39513,
1606.39906,-1074.29311,
1608.14396,-1156.30337,
1410.97036,-1170.26256),
ncol=2,byrow = T)

Arrow 5

punta5<-matrix(c(
1119.57220,-1510.51789,
1327.21520,-1383.14025,
1533.11329,-1386.63005,
1529.62350,-1465.15051,
1639.55214,-1564.60976,
1637.80724,-1618.70164,
1349.89889,-1625.68123),
ncol=2,byrow = T)

\end{verbatim}
\normalsize

\

\noindent
{\bf Comparison between Arrows 1 and 3}

After removing the effect of  {\bf Size} (scale),  {\bf Location} (translation) and {\bf Orientation} (rotation)  to standardize the individuals, we match them at the common centroid obtaining the polygons of Fig.~\ref{foto3}.

\begin{figure}[b]
\sidecaption
\includegraphics[scale=.4]{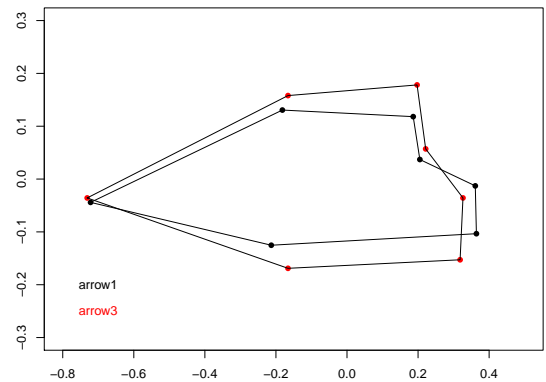}
\caption{Polygons of  arrows 1 and 3}
\label{foto3}
\end{figure}

The minimum {\it Residual Distance} between configurations (arrows), i.e., the value of the
{\it Procrustes statistic} for testing the null that ``No significance differences exist between arrows 1 and 3'' is $0.01567681$:

\small
\begin{verbatim}

> tamapunta1<-sqrt(sum(apply(punta1,2,var)*(7-1)))
> spunta1<-scale(punta1/tamapunta1,scale=F)
> tamapunta3<-sqrt(sum(apply(punta3,2,var)*(7-1)))
> spunta3<-scale(punta3/tamapunta3,scale=F)
> library(shapes)
> (procdist(spunta1,spunta3,type="partial"))^2
[1] 0.01567681

> sd(spunta1-spunta3)
[1] 0.03502222                                                             (1)

\end{verbatim}
\normalsize

Because of {\tt (1)}, choosing $\eta= 0.03502222$, we shall obtain a standard normal distribution for $\, (X_2- X_1)/ \eta \, $ and hence

 $$G(X_1,X_2)  \approx  \eta^2 \, \chi^2_g$$

\noindent
being $g=2p-4 = 10$. Then, the p-value of this classical test  will be

$$
P(Procrus. Stat. > 0.01567681 ) =  P(\chi^2_{10}>  0.01567681/(0.03502222^2))$$

$$ =  P(\chi^2_{10} > 12.78116)=
 1-pchisq(12.78116,10)=0.2361661$$

\noindent
accepting the null hypothesis of no significance differences between Arrows 1 and 3.

Nevertheless,  using the Mahalanobis distance we can conclude that the errors do not follow a multivariate normal distribution,

\small
\begin{verbatim}

> dipuntas2<-mahalanobis(spunta1-spunta3,colMeans(spunta1-spunta3),
+ var(spunta1-spunta3))
> ks.test(dipuntas2,"pchisq",7)

       One-sample Kolmogorov-Smirnov test

data:  dipuntas2
D = 0.76612, p-value = 8.265e-05
alternative hypothesis: two-sided

\end{verbatim}
\normalsize

Hence, to assume a common $\eta$ for all the $c_{ij}$ such that

$$\frac{X_2 - X_1}{\eta} \leadsto N(0,1)$$

\noindent
 is unrealistic. It is better to consider a model

$$0'9 N(0,1) + 0'1 N(0,\nu)$$

\noindent
and to use the VOM+SAD approximation (\ref{approx1}),  programmed into the R function $\;$ {\tt apro3}$(g,\nu,\epsilon,t) \;$
 to compute the p-value. Obtaining from the data  $\eta = 0.04020902$ and $\nu =  0.032261$, we have

\small
$$P(Procrus. Stat. > 0.01567681 ) =  P(Procrus. Stat./(0.04020902^2) >  0.01567681/(0.04020902^2)) $$

$$=  P(Procrus. Stat./(0.04020902^2) > 22.95901)= P( G_s(X_1,X_2) >  22.95901) =
0.006318776$$
\normalsize

\noindent
because

\begin{verbatim}

> apro3(10,0.032261,0.1,22.95901)
[1] 0.006318776

\end{verbatim}

Then, because

$$(VOM+SAD) \mbox{ p-value } = 0.006318776$$

\noindent
we reject, in a more robust way, the null hypothesis of no significance differences between Arrows 1 and 3.

\

\noindent
{\bf Comparison between Arrows 5 and 6}

After removing the effect of  {\bf Size} (scale),  {\bf Location} (translation) and {\bf Orientation} (rotation)  to standardize the individuals we match them at the common centroid having the polygons of Fig.~\ref{foto4}.

\begin{figure}[b]
\sidecaption
\includegraphics[scale=.4]{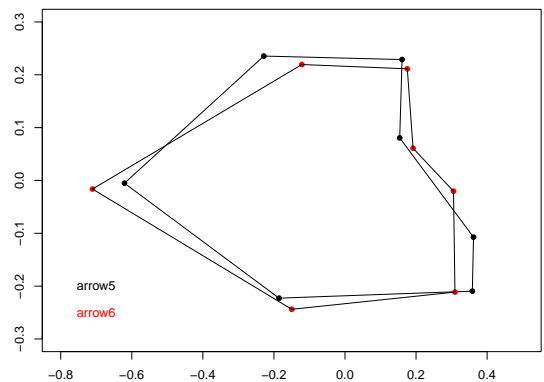}
\caption{Polygons of  arrows 5 and 6}
\label{foto4}
\end{figure}

The minimum {\it Residual Distance} between configurations (arrows), i.e., the value of the
{\it Procrustes statistic} for testing the null that ``No significance differences exist between arrows 5 and 6'' is $0.03711933$,

\small
\begin{verbatim}

> tamapunta5<-sqrt(sum(apply(punta5,2,var)*(7-1)))
> spunta5<-scale(punta5/tamapunta5,scale=F)
> tamapunta6<-sqrt(sum(apply(punta6,2,var)*(7-1)))
> spunta6<-scale(punta6/tamapunta6,scale=F)
> library(shapes)
> (procdist(spunta5,spunta6,type="partial"))^2
[1] 0.03711933

> sd(spunta5-spunta6)                                                      (2)
[1] 0.05343598

\end{verbatim}
\normalsize

Because of {\tt (2)}, if  $\eta= 0.05343598$  we shall obtain a standard normal distribution for $\, (X_2- X_1)/ \eta \, $. Then,

 $$G(X_1,X_2)  \approx  \eta^2 \, \chi^2_g.$$

\noindent
and hence, the p-value of this classical test  will be

$$
P(Procrus. Stat. > 0.03711933 ) =  P(\chi^2_{10}>   0.03711933/(0.05343598^2)) $$

$$=  P(\chi^2_{10}  > 12.99968)=
 1-pchisq(12.99968,10)=0.2236897$$

\noindent
accepting the null hypothesis of no significance differences between Arrows 5 and 6.

Nevertheless, using the Mahalanobis distance we can conclude that the errors do not follow a multivariate normal distribution,

\small
\begin{verbatim}

> dipuntas<-mahalanobis(spunta6-spunta5,colMeans(spunta6-spunta5),
+ var(spunta6-spunta5))
> ks.test(dipuntas,"pchisq",7)

        One-sample Kolmogorov-Smirnov test

data:  dipuntas
D = 0.76677, p-value = 8.093e-05
alternative hypothesis: two-sided

\end{verbatim}
\normalsize

\noindent
Then, to assume a common $\eta$ for all the $c_{ij}$ such that

$$\frac{X_2 - X_1}{\eta} \leadsto N(0,1)$$

\noindent
 is unrealistic. It is better to consider a model

$$0'9 N(0,1) + 0'1 N(0,\nu)$$

\noindent
and to use the VOM+SAD approximation (\ref{approx1}),  programmed as the R function $\;$  {\tt apro3}$(g,\nu,\epsilon,t) \;$
 to compute p-values. From the data we obtain $\eta = 0.06834322$ and $\nu =  0.0389347$ and hence,

\small
$$P(Procrus. Stat. > 0.03711933 ) =  P(Procrus. Stat./(0.06834322^2) >  0.03711933/(0.06834322^2)) $$

$$=  P(G_s(X_1,X_2)   > 7.947111 )=
0.5405565$$
\normalsize

\noindent
because

\begin{verbatim}

> apro3(10,0.0389347,0.1,7.947111)
[1] 0.5405565

\end{verbatim}

Then, because

$$(VOM+SAD) \mbox{ p-value } = 0.5405565$$

\noindent
we finally accept, in a more robust way,  the null hypothesis of no significance differences between Arrows 5 and 6.
\end{example}

\section{Conclusions}

Classical Morphometric Analysis based on Landmarks is reviewed from a descriptive and inferential point of view. Because both are based on sample means and least-squares they are not robust.

We first robustify the descriptive measures proposing robust ones. Then we consider a Contaminated Normal Model distribution instead of a classical Normal one to make robust inferences. Namely, for this mixture model we obtain an von Mises approximation for the p-value of a test for the null hypothesis of
no significance differences between two individuals based on their shapes.

We also obtain a very accurate saddlepoint approximation of this von Mises approximation.  We conclude the paper with some applications using QGIS as  Geographical Information System.

\begin{acknowledgement}
This work is partially supported by Grant HAR2015-68876-P from Ministerio de Econom\'{\i}a y Competitividad (Spain).
\end{acknowledgement}

\bibliographystyle{spbasic}

\end{document}